\documentclass[letterpaper,11pt]{article}
\usepackage{microtype}
\usepackage[letterpaper,margin=1in]{geometry}
\usepackage[UKenglish]{babel}
\usepackage[numbers,sort&compress]{natbib}
\usepackage{color}
\usepackage{doi}
\usepackage{latexsym}
\usepackage{amsmath}
\usepackage{amssymb}
\usepackage{amsthm}
\usepackage[all,warning]{onlyamsmath}
\usepackage{enumitem}
\usepackage{xparse}
\usepackage{mathtools}
\usepackage{graphicx}
\usepackage[basic]{complexity}
\usepackage[small]{caption}
\usepackage{wrapfig}
\usepackage{hyperref}

\theoremstyle{plain}
\newtheorem{theorem}{Theorem}
\newtheorem*{theorem-bis}{Theorem}
\newtheorem{lemma}[theorem]{Lemma}

\newtheorem{corollary}[theorem]{Corollary}
\newtheorem{claim}[theorem]{Claim}

\theoremstyle{definition}\synctex=1

\newtheorem{conjecture}[theorem]{Conjecture}

\theoremstyle{remark}
\newtheorem{remark}[theorem]{Remark}

\setlist[itemize]{label=--}
\setlist[enumerate]{label=(\arabic*),labelindent=\parindent,leftmargin=*}

\DeclarePairedDelimiter\braces{\{}{\}}
\NewDocumentCommand\set{O{}mg}{\ensuremath{\braces[#1]{#2\IfNoValueTF{#3}{}{\,:\,#3}}}}

\newclass{\lcl}{LCL}
\newclass{\lco}{LCO}
\newclass{\local}{LOCAL}
\newclass{\slocal}{SLOCAL}

\newcommand{\LL}{\mathcal{L}}

\newcommand{\Sigcc}{\Sigma^{\operatorname{cc}}}
\newcommand{\Picc}{\Pi^{\operatorname{cc}}}

\newcommand{\ns}{\textsc{None selected}}

\newcommand{\alos}{\textsc{Alos}}
\newcommand{\leader}{\textsc{Leader election}}

\newcommand{\namedref}[2]{\hyperref[#2]{#1~\ref*{#2}}}

\newcommand{\colambda}[1]{\operatorname{co-\Lambda}_{#1}}

\newenvironment{myabstract}
               {\list{}{\listparindent 1.5em%
                        \itemindent    \listparindent
                        \leftmargin    1cm
                        \rightmargin   1cm
                        \parsep        0pt}%
                \item\relax}
               {\endlist}

\newenvironment{mycover}
               {\list{}{\listparindent 0pt
                        \itemindent    \listparindent
                        \leftmargin    1cm
                        \rightmargin   1cm
                        \parsep        0pt}%
                \raggedright
                \item\relax}
               {\endlist}

\newcommand{\myemail}[1]{\,$\cdot$\, {\small #1}}
\newcommand{\myaff}[1]{\,$\cdot$\, {\small #1}\par\medskip}

\hypersetup{
  colorlinks=true,
  linkcolor=black,
  citecolor=black,
  filecolor=black,
  urlcolor=[rgb]{0,0.1,0.5},
  pdftitle={Local verification of global proofs},
  pdfauthor={}
}

\begin{document}

\newgeometry{margin=1in,bottom=1in}

\begin{mycover}
{\huge\bfseries\boldmath Local verification of global proofs \par}
\bigskip

\textbf{Laurent Feuilloley}
\myemail{feuilloley@irif.fr}
\myaff{IRIF, CNRS and University Paris Diderot}

\textbf{Juho Hirvonen}
\myemail{juho.hirvonen@cs.uni-freiburg.de}
\myaff{University of Freiburg}

\medskip
\begin{myabstract}
\noindent\textbf{Abstract.} In this work we study the cost of local and global proofs on distributed verification. In this setting the nodes of a distributed system are provided with a nondeterministic proof for the correctness of the state of the system, and the nodes need to verify this proof by looking at only their local neighborhood in the system.

Previous works have studied the model where each node is given its own, possibly unique, part of the proof as input. The cost of a proof is the maximum size of an individual label. We compare this model to a model where each node has access to the same global proof, and the cost is the size of this global proof.

It is easy to see that a global proof can always include all of the local proofs, and every local proof can be a copy of the global proof. We show that there exists properties that exhibit these relative proof sizes, and also properties that are somewhere in between. In addition, we introduce a new lower bound technique and use it to prove a tight lower bound on the complexity of reversing distributed decision and establish a link between communication complexity and distributed proof complexity.
\end{myabstract}

\vfill

The authors received additional support from ANR project DESCARTES. The first author receives additional support from INRIA project GANG. The second author was supported by Ulla Tuominen Foundation.

\thispagestyle{empty}
\setcounter{page}{0}
\newpage
\restoregeometry

\end{mycover}

\section{Introduction}\label{sec:introduction}

In distributed decision a distributed system must decide if its own state satisfies a given property. When compared to classical decision problems, the crucial difference is that each node of the distributed system must make its own local decision based only on information available in its local neighborhood. We say that the system \emph{accepts} if all nodes accept, and otherwise the system \emph{rejects}.

A distributed system is modeled as a communication graph where edges denote nodes that the connected nodes can directly communicate with each other. Distributed decision where each node gets to see its constant-radius neighborhood in the graph is called \emph{local decision}~\cite{FraigniaudKP11}. It is possible, as the name suggests, to decide local properties of the system, for example whether a given coloring is correct. On the other hand, it is impossible to decide global properties like whether a given set of edges forms a spanning tree.

To decide global properties, nodes can be provided with a nondeterministic proof~\cite{KormanKP05}. Each node gets its own proof string as an additional input. Nodes can gather all input in their constant-radius neighborhood, including these local proof strings, and decide to accept or reject. We say that such a scheme decides a property if there exists an assignment of local proofs to make all nodes accept if and only if the system satisfies the required property.

For example, to prove that a given set of edges forms a spanning tree, each node can be provided with the name of a designated root of the tree, and its distance to the root. Nodes can check that they agree on the identity of the root, and that they have exactly one neighbor with smaller distance to the root along the edges of the spanning tree.

We are interested in minimizing the size of the proof. In particular, we want to minimize the size of the largest label given to a single node. The local proofs often contain redundant information between the different local proof strings. For example, in the previous case each node must know the name of the root. The distances to the root are also highly correlated between neighbors. We approach this question by comparing the size of local proofs to \emph{global} proofs. In this setting each node has access to the same, universal proof string. The decision mechanism remains otherwise the same.

It is easy to see that minimum sizes of global and local proofs bound each other. A global proof can simply be a list of the local proof strings. Conversely, a local proof can copy the same global proof for each node. In this work we study how these two proof sizes relate to each other for different properties.

Unlike in the centralized setting, distributed decision cannot be reversed trivially. This is due to the fact that the distributed decision mechanism is asymmetric: all nodes must accept a correct input, but a failure might only be detectable locally. Decision can be reversed using a logarithmic number of additional nondeterminism~\cite{GoosS16,FeuilloleyFH16}: such a proof involves constructing a spanning tree that points to an error. This is an even more general primitive in distributed proofs: the proof must convince the nodes that a local defect exists somewhere in the graph, and only the nodes that are located close to this defect can verify it. We show that the existing upper bounds are asymptotically tight: reversing decision requires a local proof of logarithmic size.

Since the distributed verification of a proof happens locally, a distributed proof of a global property must carry information between distant parts of the input graph. This has led to the use of lower bound techniques from communication complexity for distributed decision. On the other hand proving lower bounds inside the nondeterministic hierarchy of local decision~\cite{FeuilloleyFH16} with multiple levels of nondeterminism seems to be hard. This is partially due to the fact that current lower bound techniques from communication complexity cease to work. We formalize this intuition by establishing a connection between the nondeterministic local decision hierarchy and the nondeterministic communication complexity hierarchy~\cite{babai86complexity}.

\paragraph{Motivation}
The first proof-labeling schemes were designed in the context of self-stabilizing algorithms, where a distributed algorithm would run on the graph, and would, in addition to the output, keep some information to verify that the state of the network is not corrupted. Similar scenarios exist for global proofs. For example, one may consider a network where the machines compute in a distributed fashion, but an external operator with a view of the whole network can once in a while broadcast a piece of information, such as the name of a leader. As one expects this type of update to be costly, the focus is on minimizing the size of such broadcast information.

Our research belongs to a recent line of work that establishes the foundations of a theory of complexity for distributed network computation. In this context, the certificate comes from a prover, and one studies the impact of non-determinism on computation and the minimal amount of information needed from the prover to decide a task. Global proofs are a natural alternative form of non-determinism. Moreover, in proof-labeling schemes a part of the certificate is often global. For example, the name of a leader is given to all nodes. Global proofs can be used to study how much of such redundant information a local proof must have. Finally, one may consider that global proofs are the most natural equivalent of classical non-determinism: only the algorithm is distributed and we ask what is the cost if distributing the proof.

\paragraph{Related work}
Proof-labeling schemes have been defined in \cite{KormanKP05, KormanKP10}. An important result in the area is the tight bound on the size of the certificates for certifying minimum spanning tree \cite{KormanK07}. Recently, several variations have been defined, for verifying approximation \cite{Censor-HillelPP17}, with non-constant verification radius \cite{OstrovskyPR17}, with a dependency between the number of errors and the distance to the language \cite{FeuilloleyF17}, and variations on the communication model \cite{Patt-ShamirP17}. An analogue of the polynomial hierarchy for distributed decision has also been defined \cite{FeuilloleyFH16}.

Another line of work uses a slightly different notion of non-determinism. Fraigniaud et al.\ \cite{FraigniaudKP11} consider a similar kind of scheme with a prover and local verifier, but with the constraint that the certificates should not depend on the identifiers of the network. For these works, and more generally the complexity theory of distributed decision, we refer to a recent survey~\cite{FeuilloleyF16}.

The idea of a prover for computation in a network, or in a system with several computational units, appears outside of distributed computing, and usually with a global proof. In property testing, models where a prover provides a certificate to the machine that queries the graph have been considered \cite{LovaszV13, GurR15}.
In two-party communication complexity, non-determinism comes as a global proof that both players can access. Along with non-determinism the authors of \cite{babai86complexity} define a hierarchy. Separating the levels of this hierarchy is still a major open problem \cite{GoosP016}.

\paragraph{Our contributions.}
We formalize the notion of \emph{global proofs} for nondeterministic local verification. We study them, in particular comparing the global and local proof complexities of distributed verification.

One main goal of this line of research is to understand the price of locality in nondeterministic distributed verification --- that is, how much information must be repeated in the local proofs of the nodes in order to allow local decision of global problems.
\begin{enumerate}
  \item We show that the price of locality can exhibit the extreme possible values. An example of a maximally \emph{global} property for distributed verification is the language where at most node is selected. This is one of the core primitives in distributed verification: proving that at most one event of a given type happens in the whole graph. On the other hand, we show that when verifying that at least one node is selected, a global proof must use enough bits to essentially copy every local proof label.
  \item We introduce a new proof technique for proving lower bounds for local verification. This proof technique is based on analyzing the neighborhood graph labeled with the local proofs. We use it to show that reversing decision requires $\Omega(\log n)$-bit local proof and a $\Omega(n \log n)$-bit global proof. Our proof technique is somewhat similar to the one used by G\"{o}\"{o}s and Suomela to prove their $\Omega(\log n)$ lower bounds for local proofs~\cite{GoosS16}. Their proof technique relies on combining two fragments of \emph{yes}-instances to produce an accepting \emph{no}-instance. This is not sufficient for our results, since we want to prove lower bounds for languages for which two fragments of \emph{yes}-instances joined together might still produce a \emph{yes}-instance.
  \item We establish a connection between nondeterministic verification and nondeterministic communication complexity. Proving separations for the hierarchy of nondeterministic communication complexity has been an open question since its introduction over 30 years ago~\cite{babai86complexity}. We show that proving similar separations for the hierarchy of nondeterministic local decision is connected to this question: for every boolean function $f$ we construct a distributed language such that it can be decided on the $k$th level if $f$ can be decided on the $k$th level of the communication complexity hierarchy. For global proofs, this can be strengthened to show that verification schemes also imply communication protocols. This formalizes the previous intuition that proving lower bounds for nondeterministic local verification is potentially hard as it would imply proving lower bounds for nondeterministic communication complexity.
\end{enumerate}

\section{Model and definitions}\label{sec:model}

The network is modeled by a simple graph with no loops. The size of the graph is denoted by~$n$. The nodes are given distinct identifiers in a range that is polynomial in $n$, that is, IDs on $O(\log n)$ bits.

\paragraph{Distributed decision.}

A \emph{distributed language} is a set of graphs, whose nodes and edges can have inputs. Distributed languages are often assumed to be computable (from the centralized computing perspective), but this is irrelevant for the current paper. An example is the language \textsc{spanning tree}, which is the set of graphs whose edges are labeled with 1 or 0, such that edges labeled with 1 form a spanning tree of the graph.

A \emph{local decision algorithm} with radius $t$, is a local algorithm in which every node $v$ first gathers all the information about its $t$-neighborhood (the structure of the graph, the IDs of the nodes, the local inputs), and possibly some proofs given by one or several provers, and outputs a decision, accept or reject, based on this information. The distance $t$ is constant with respect to the size of the network $n$. The verifier is uniform \emph{i.e.} it does not know the size of the graph.

In a basic \emph{local decision scheme}, there is no prover. We say that the scheme decides a language, if for every labeled graph: all the nodes accept, if and only if, the labeled graph belongs to the language. In a \emph{non-deterministic scheme}, the following should hold: there exists a proof that makes every node accept, if and only if, the language belongs to the language.

\paragraph{Different types of proofs.}
A \emph{(purely) local proof}, the prover provides every node with its own certificate. For such a scheme, the local proofs have the same size which depends only on the language and on the size of the network $n$. The size of the proof is denoted by $s_\ell(n)$. This is the classic framework of proof-labeling schemes. We introduce \emph{(purely) global proofs}, where the prover provides a single certificate, and every node can access it. Its size is denoted by~$s_g(n)$. Finally, in \emph{mixed proofs}, the prover provides a global proof and local proofs. The size, denoted by $s_m(n)$, is the sum of the size of the global proof, and the size of the concatenation of the local proofs.

\paragraph{Price of locality.}
We define a \emph{Price of Locality} for proofs, by analogy with the \emph{Price of Anarchy} in algorithmic game theory  \cite{KoutsoupiasP09, Papadimitriou01}. Note that this is not the same as the price of locality that appears in the title of \cite{Kuhn2005}.
The price of locality (PoL) of a language $\LL$ is defined as the ratio between the size of the concatenation of the purely local certificate, divided by the size of the mixed certificate. That is: \[PoL(n) =\frac{n\cdot s_\ell(n)}{s_m(n)}.\]

\section{The price of locality}\label{sec:price}

In this section, we study the size of local, mixed and global proofs for different problems, and the price of locality that follows. 

\subsection{Proof sizes}

In this subsection some general inequalities between the sizes of the different proof sizes are proven. We then discuss the definition of the price of locality.

\begin{theorem}\label{thm:inequalities}
For any language, the optimal proof sizes respect the following inequalities.
\begin{align}
s_\ell(n) &\leq s_m(n) \leq s_g(n)\\
s_m(n) &\leq n\cdot s_\ell(n)\\
s_g(n) &\leq n\cdot s_\ell(n) + O(n\log n).
\end{align}
\end{theorem}

\begin{proof}
The first line of inequalities mainly follows from the definitions. Suppose one is given a mixed certificate for a language, with local certificates of size $f(n)$ each, and a global certificate of size $g(n)$.
The size of this mixed certificate is $s_m(n)=n \cdot f(n)+g(n)$. Then one can create a local proof of size $f(n)+g(n)$, by giving to every node its local part concatenated with the global part. Thus $s_\ell(n) \leq s_m(n)$.
The inequality $s_m(n) \leq s_g(n)$ holds because the mixed proof is a generalization of the global proof.
Similarly, if there exists local certificates of size $s_\ell(n)$, then one can use them in the mixed model. The size measured in the mixed model will then be $n\cdot s_\ell(n)$.
Finally, given local certificates, one can craft a global certificate. The global certificate consists in a list of couples, each couple being an ID and the local certificate of the node with this ID.  The size is in $n\cdot s_\ell(n) + O(n\log n)$ because identifiers are on $O(\log n)$ bits.
\end{proof}

For a given language, the price of locality is defined as the ratio $n\cdot s_\ell(n)/s_m(n)$, that is the size of the concatenation of all the optimal local proofs divided by the size of the optimal global proof. The inequalities above insure that with this definition the ratio is between 1 and $n$. Note that if we had defined the ratio with global proofs instead of mixed ones, a priori the price could be smaller than 1 for some languages, thus not a price per se. This is because, unlike mixed proofs, global proofs are not a generalization of local proofs. Section \ref{sec:bipartite} gives an example where this happens. However, for the remaining of this section, we show upper bounds on global proofs, because these are stronger, and the bounds on the price of locality follow.

\subsection{High price of locality}

In this section we prove that global proofs can be much more efficient than local proofs. In other words, it can be very costly to distribute the proof.
We consider the language \textsc{At-most-one-selected} (\textsc{Amos}), that has been defined and used in~\cite{FraigniaudKP13}. In this problem, the nodes are given binary inputs, and the \emph{yes}-instances are the ones such that at most one node has input 1.

\begin{theorem}
The global proof size for \textsc{Amos} is in $O(\log n)$.
\end{theorem}

\begin{proof}
The prover strategy on \emph{yes}-instances is the following. If there is exactly one selected node, the prover provides the ID of the node as the global certificate, otherwise it provides an empty label. The verification algorithm is, for every node $v$: if $v$ is selected and the certificate is not its~ID, then reject, otherwise accept.
It is easy to check that this scheme is correct. First, if no node is selected, all nodes accept, for all certificate. Second, if one nodes is selected, then the prover provides its ID as a certificate, and thus the selected node accepts, and all the other nodes too. Finally, if two or more nodes are selected, at most one of them has its ID written in the global certificate, because the IDs are distinct~; thus at least one node is rejecting.
\end{proof}

In \cite{GoosS16}, the authors prove that verifying that exactly one node is selected requires $\Omega(\log n)$ local certificate. The proof basically shows that without this amount of proof, an instance with two leaders would be accepted. This reasoning holds for \textsc{Amos}, and we can derive a $\Omega(\log n)$ lower bound for local certificates as well. The follows that the price of locality is as large as it can be, that is, order of $n$.

\begin{corollary}
The price of locality for \textsc{Amos} is in $\Theta(n)$.
\end{corollary}

Other examples of problems with high price of locality are the languages that are very non-local. A general upper bound on local proof size is $O(n^2)$\cite{KormanKP10} (if the inputs have constant size). This is because the prover can always encode the whole graph in the certificate with $O(n^2)$ bits, and a local verification is enough to ensure that the graph described in the certificate is the same as the communication graph. For this kind of proof the sum of the sizes of the local proof is $O(n^3)$, whereas the associated global proof is $O(n^2)$. Several languages are proven to have no better local proof than this one, including the language of graphs that have no non-trivial automorphism, and, up to logarithmic terms, the language of non-3-colorable graphs \cite{GoosS16}.

\subsection{Intermediate price of locality}

In this section, we show that having a global certificate helps saving space for the well-studied language \textsc{Minimum Spanning Tree} (\textsc{Mst}). In this case, the price of locality is $\Theta(\log n)$, thus intermediate between $n$ (the previous subsection), and constant (the next section).

The language  \textsc{Mst} is the set of weighted graphs, in which a set of edges is selected, and forms a minimum spanning tree of the graph. The edge weights are assumed to be polynomial in $n$, thus they can be written on $O(\log n)$ bits. For simplicity we suppose that all weights are distinct.

In \cite{KormanK07}, the authors show that there exist local proofs of size $O(\log^2n)$ for \textsc{Mst}, and that this bound is tight. We show a simple global proof that has size $O(n\log n)$. As a mixed proof for the simpler language \textsc{Spanning Tree} requires $\Omega(n\log n)$ (see Section \ref{sec:alos}), this bound is tight. 

\begin{theorem}\label{thm:mst}
The global proof size for \textsc{Minimum Spanning Tree} is in $O(n\log n)$.
\end{theorem}

\begin{proof}
We first describe the scheme.
On a \emph{yes}-instance the prover provides a list of the selected edges with their weights. This global certificate has size $O(n \log n)$.

We now describe the verification algorithm. Every node first will checks that the certificate is correct regardless of the graph. That is, every node will check that:
\begin{itemize}
\setlength\itemsep{1pt}
\item The certificate is a well-formed edge list. Let $L$ be this.
\item The list $L$ describes an acyclic graph. That is that there is no set of nodes $w_1, w_2, ..., w_k$ such that $(w_1,w_2), (w_2,w_3),..., (w_{k-1},w_k)$, and $(w_k, w_1)$ appear in the list.
\item The list $L$ describes a connected graph. That is for any pair of nodes present in the list, there exists a path in the list that connects them.
\end{itemize}

Then every node $v$ of the graph checks locally that:
\begin{itemize}
\setlength\itemsep{1pt}
\item The $L$ is consistent with the edges that are adjacent to $v$.
\item The nodes $v$ has an adjacent selected edge.
\item For every $e=(v,w)$ in $E\setminus L$, and every edge $e'$ on the path from $v$ to $w$ in $L$, the weight of $e'$ is smaller than the weight of $e$.
\end{itemize}

We now prove the correctness of the scheme. The first part of the verification insures that the set of edges described by $L$ forms an acyclic connected graph. The two first checks of the second part insure that it is spanning the graph, and that it contains the selected edges. As it is a spanning tree, it must then be exactly the set of selected edges. Finally, remember that the so-called \emph{cycle property} states that a spanning tree verifying the last item of the previous algorithm is minimal~\cite{CormenLRS09}.
\end{proof}


\section{Locality for free and reversing decision} \label{sec:alos}

In this section, we show that for some languages there exists local proofs of size $O(\log n)$ and that any mixed proof has size $O(n \log n)$. It follows that in this case, the price of locality is constant, that is the locality of the proofs comes for free.

The language we consider, called \textsc{At least one selected} (\alos), consists of all labeled graphs such that at least one node has a non-zero input label. We say that a node with a non-zero input label is \emph{selected}. Proving that at least one node has some special property (being the root, having some intput, being part of some special subgraph) is an important subroutine in many scheme. 

On a more fundamental perspective, reversing decision basically deals with proving that some node is rejecting, which falls into the scope of the \alos. It has long been known that $O(\log n)$ local proof is sufficient for reversing decision, and the current section shows that not only this is optimal, but also one cannot gain by using global proofs.

\begin{theorem} \label{thm:alos}
A mixed proof for the language \alos{} requires $\Omega(n \log n)$ bits.
\end{theorem}

The theorem is equivalent to state that the language requires either $\Omega(\log n)$ bits per local proof or $\Omega(n \log n)$ bits of global proofs. Now we are ready to begin the proof of Theorem~\ref{thm:alos}.

\begin{proof}[Proof of Theorem~\ref{thm:alos}]
Consider a mixed scheme with local certificates of size $f(n)$ and a global certificate of size $g(n)$. Let $r$ be the verification radius of the scheme.

\paragraph{Blocks.}
The lower bound instances are consistently oriented cycles of length at most $n=(b+1)(2r+1)$, for some integer $b$. Cycles are constructed from blocks of $2r+1$ nodes: the $i$th block is a path $B_i = (v_j, v_{j+1}, \dots, v_{j+2r})$, where $j = i(2r+1)+1$, oriented consistently from $v_j$ to $v_{j+2r}$. Each node $v_j$ is labeled with the unique identifier $j$.

\paragraph{Constructing instances from blocks.}

Let $\pi \colon [b] \to [b]$ be a permutation on the set of the first $b$ blocks. Each permutation defines a cycle $C_{\pi}$ where we take the blocks in the order given by $\pi$, and finally take the $(b+1)$th block. Each pair of consecutive blocks in $\pi$ is connected by an edge, and $B_{b+1}$ is connected to $B_{\pi^{-1}(1)}$.

Finally, the center node $v_{b(2r+1)+r+1}$ of $B_{b+1}$ is labeled with a non-zero label, making the instance a \emph{yes}-instance. All other nodes are labeled with the zero-label. Denote this family of permuted \emph{yes}-instances by $\mathcal{C} = \{ C_{\pi} \}_{\pi}$.

\paragraph{Labeled blocks.}

The prover assigns a local proof of $f(n)$ bits to each node. Thus, there are $2^{f(n)(2r+1)}$ different labeled versions of each block. We call these \emph{labeled blocks}. Denote by $B_{i,\ell}$ the block $B_i$ labeled according to $\ell$. We call $B_i$ the \emph{type} of $B_{i,\ell}$

Consider two labeled blocks, $B_{i,\ell}$ and $B_{j,\ell'}$, in this order, linked by an edge. We say that labeled blocks are accepting from $B_{i,\ell}$ to $B_{j,\ell'}$ with global certificate $L$ if, when we run the verifier on the nodes that are at distance at most $r$ from an endpoint of the connecting edge, all these nodes accept. We denote this by $B_{i,\ell} \rightarrow_L B_{j,\ell'}$.

For each choice $L$ of the global certificate, this edge relation defines a graph $G_{\mathcal{B}, L}$ on the set of labeled blocks. A path in $G_{\mathcal{B},L}$ corresponds to a labeled path fragment in which all nodes at least $r$ steps away from the path's endpoints accept. Finally, an accepting cycle is a cycle in $G_{\mathcal{B},L}$ such that all nodes accept.

\paragraph{Bounding the overlap of certificates.}

For each $C_{\pi} \in \mathcal{C}$, there must exist an accepting assignment of certificates to the nodes. Let $L$ denote the global part of this accepting certificate. Such a $C_{\pi}$ corresponds to a directed cycle in $G_{\mathcal{B},L}$. Note that in this cycle the last edge can be omitted as it would always link the last block to the first block. Then $C_{\pi}$ corresponds to a directed path $P(C_{\pi},L)$ of length $b$ in $G_{\mathcal{B},L}$. Denote the set of labeled blocks on this path by $S(C_{\pi},L)$.

Let $\mathcal{C}_L$ denote the set of instances such that there exists an accepting local certification given the global certificate $L$. Every \emph{yes}-instance has an accepting certification, so there must exist $L^*$ with
\[
|\mathcal{C}_{L^*}| \ge |\mathcal{C}| / 2^{g(n)}.
\]

Now consider any two instances $C_{\pi}$ and $C_{\pi'}$ in $\mathcal{C}_{L^*}$. We drop the specification of the global certificate from the notation. Assume that $C_{\pi}$ and $C_{\pi'}$ use the same set of blocks, that is $S(C_{\pi},L^*) = S(C_{\pi'}, L^*)$. Also assume without loss of generality, that $\pi$ is the identity permutation. Now in $P(C_{\pi'})$ there must exist a \emph{back edge} with respect to $\pi$, that is, an edge between labeled blocks $B$ and $B'$, of types $B_{\pi'^{-1}(i)}$ and $B_{\pi'^{-1}(i+1)}$ respectively, such that $\pi'^{-1}(i) > \pi'^{-1}(i+1)$. This is because we assumed that the instances consist of the same blocks, but are different. Therefore at some point an edge of $C_{\pi'}$ must go backwards in the order of $\pi$. We also have that $B,B' \neq B_{b+1}$ as if there is no back edge before reaching $B_{b+1}$, we must have $C_{\pi} = C_{\pi'}$.

This implies that there is an accepting cycle formed by taking first the path from $B$ to $B'$ along $P(C_{\pi})$ and then an edge from $B'$ to $B$. This cycle does not contain a selected node. It follows that there is a \emph{no}-instance of size at least $2(2r+1)$ and a certification that causes the verifier to accept the instance. Therefore we have the following lemma.

\begin{lemma} \label{lem:distinct-sets}
  For all pairs of instances $C_{\pi}$, $C_{\pi'}$ with the same accepting global certificate $L$, we have that $S(C_{\pi},L) \neq S(C_{\pi'},L)$.
\end{lemma}

\begin{remark}
Note that the contradicting instances can be of size $2(2r+1)$ but the identifiers can be of size $n$ and the certificates of size $f(n)$. Therefore the lower bound only holds for uniform verifiers that do not get any guarantees except that 1) the identifiers come from the set $[n+c]$, for some constant $c$, and 2) the certificates are of size at least $f(n)$.

Alternatively it is possible to consider \alos{} on possibly disconnected instances so that every connected component must have at least one node selected. In this case the proof will fool even a non-uniform prover.
\end{remark}

\paragraph{Counting argument.} By Lemma~\ref{lem:distinct-sets}, each pair of permutations $\pi, \pi'$ in $\mathcal{C}_{L^*}$ must induce a different set of labeled blocks that form the accepting certifications of instances $C_{\pi}$ and $C_{\pi'}$. The number of different permutations in $\mathcal{C}_{L^*}$ is at least $b! / 2^{g(n)}$. On the other hand, the number of different sets of labeled blocks, selecting a block of each type, is $2^{f(n)(2r+1)b}$. As shown in Lemma~\ref{lem:distinct-sets}, to have a legal certification, we must have that $2^{f(n)(2r+1)b + g(n)} \ge b!$.

Using Stirling's approximation we get that $f(n)(2r+1)b + g(n) \ge b \log_2 b - (\log_2e) b + O(\ln b)$. Since $b = \Theta(n)$ and $r=O(1)$, this implies that either $f(n) = \Omega(\log n)$ or $g(n) = \Omega(n \log n)$. Thus the mixed proof has size $\Omega(n \log n)$ \qedhere
\end{proof}

\begin{corollary} \label{cor:inversion}
    Inverting decision with a single additional level of nondeterminism requires local certificates of size $\Omega(\log n)$ or global certificates of size $\Omega(n \log n)$.
\end{corollary}

\begin{proof}
Consider the language \ns{}, that is, the language of labeled graphs such that all nodes have the zero label. Clearly this language is in $\Lambda_0$ and \alos{} is its complement. Finally, by Theorem~\ref{thm:alos}, deciding \alos{}, that is, reversing the decision of \ns{} requires local certificates with $\Omega(\log n)$ bits or global certificates with $\Omega(n \log n)$ bits.
\end{proof}

The proof can be adapted to several other problems, namely \leader{} (the set of graphs where exactly one node is selected), \textsc{spanning tree}, and \textsc{odd-cycle} (the set of cycles of odd length), providing a lower bound for mixed proof systems for these languages. 

\begin{corollary} \label{cor:le}
  Any mixed proof system for \leader{} requires local certificates of size $\Omega(\log n)$ or global certificates of size $\Omega(n \log n)$.
\end{corollary}

\begin{proof}[Proof of Corollary~\ref{cor:le}]
  Consider the proof of Theorem~\ref{thm:alos}. The family $\mathcal{C}$ of \emph{yes}-instances for \alos{} is also a family of \emph{yes}-instances for \leader{}. Since \leader{} $\subsetneq$ \alos{}, the proof of Theorem~\ref{thm:alos} produces \emph{no}-instances of \leader{} that the verifier accepts.
\end{proof}

\begin{corollary}\label{cor:ST}
Any mixed proof system for \textsc{spanning tree} requires local certificates of size $\Omega(\log n)$ or global certificates of size $\Omega(n \log n)$.
\end{corollary}

\begin{proof}[Proof sketch]Consider two types of instances: cycles where all the edges are selected, and cycles where all edges but one are selected. The first instances are not in the language, the second are. We can rephrase this restricted problem as: there is at least one non-selected edge. Then the same type of proof as for ALOS gives the bound.
\end{proof}

\begin{corollary}\label{cor:odd-cycle}
Any mixed proof system for \textsc{odd-cycle} requires local certificates of size $\Omega(\log n)$ or global certificates of size $\Omega(n \log n)$.
\end{corollary}

\begin{proof}[Proof sketch]
The proof of Corollary \ref{cor:odd-cycle} consists in a refinement of the proof of Theorem \ref{thm:alos}. We use the same block machinery. Note that each block has of odd length by construction. Thus the instances made of blocks that are in the language are the ones that use an odd number of blocks, and the ones outside of the language have an even number of block. 

We take $b$ to be even, such that $b+1$ is odd. As in the proof of Theorem \ref{thm:alos} the (b+1)th block is special. We divide the $b$ other blocks into two sets of equal size, and assign color white to the blocks of the first set, and black to the ones the second set. The set of \emph{yes}-instances we consider is the set of cycles on the $b+1$ blocks that alternate between white and black blocks, except on the special uncolored block. The graph $G_{\mathcal{B},L}$ is defined in the same way as before, and each node inherits the color of its block. By construction, no edge links two labeled blocks of the same color. 

We claim that Lemma 6 holds in this new context. Indeed, if there is a so-called a backward edge, then there is a cycle where the special block is not present. As all the edges respect the coloring, this means that the cycle has even length, thus it is a \emph{no}-instance. We finish the proof with a similar counting argument. The number of permutations that satisfy the coloring condition is $((b/2)!)^2$. And the number of sets of labeled blocks of each type is still $2^{f(n)(2r+1)b}$. Then Stirling approximation again implies that the mixed proof size is in $\Omega(n\log n)$.
\end{proof}

A consequence of these corollaries is that all the $\Omega(\log n)$ lower bounds obtained in \cite{GoosS16} for local certificates can be lifted to $\Omega(n \log n)$ mixed proofs with our technique. However for the problem \textsc{Amos} we studied in the previous section, our technique does not work, which is consistent with the fact that an $\Omega (n \log n)$ lower bound would contradict the $O(\log n)$ upper bound we show. As already said, the technique of \cite{GoosS16} works for \textsc{Amos}, and provides the $\Omega(\log n )$ bound for local proofs. The reason our technique fails is because we show that if the certificates are too short then one can shorten the cycles taht are \emph{yes}-instances, which is not useful for $\textsc{Amos}$, as a `subinstance' of this problem is still in the language: one can only remove selected nodes. The authors of \cite{GoosS16} show that one can glue different \emph{yes}-instances together and get a configuration that is still accepted by the nodes, and for \textsc{Amos} this basically means one can glue different instance with one node selected, and then get an instance with more than one node selected, and this instance is still accepted, which rises a contradiction. Note that because of this duality, the proof technique of \cite{GoosS16} could not help to get lower bound for \alos, even when looking only at local proofs.

It is also worth noting that the intersection of the languages \textsc{Amos} and \textsc{Alos}, is \textsc{Leader Election}. For this language, it has long been known that a PLS has size $\Theta(\log n)$, and is formed by the certificates of a spanning forest, along with the ID of the leader given to all the nodes. The results of the current and previous sections show that this decomposition is somehow mandatory: one basically needs a global part of size $\Theta(\log n)$, and a local part of size $\Theta(\log n)$.

\section{Beyond free locality}\label{sec:bipartite}

The language \textsc{Bipartite} is the set of bipartite graphs. Local proofs of constant size exist for this language: the prover can just describe a 2-coloring of the cycle by giving a bit to each node. We conjecture that for this language, even when restricted to cycles, global proofs are larger than the sum of the local proof sizes, which is $O(n)$.

\begin{conjecture}\label{conj:even-cycles}
For \textsc{Bipartite}: $s_g(n) \in \Theta(n \log n)$.
\end{conjecture}

We cannot prove the lower bound of the conjecture, but we can prove weaker inequalities. As the next theorem shows, the range of IDs is important for this problem. To study the dependency on the identifiers range, we now consider the maximum identifier to be a parameter $M$, that we do not bound by a polynomial any more. The following theorem holds.

\begin{theorem}\label{thm:bipartite} For \textsc{Bipartite}, for IDs bounded by $M$, there exist two constants $\alpha$ and $\beta$ such that
\[\alpha \max \{n, \log \log M\} \leq s_g(n) \leq \beta\min\{M,n\log M\}.\]
\end{theorem}

Note that if $M=n$ then we get a tight $\Theta(n)$ bound, and that the $n\log M$ upper bound can be derived from the proof of Theorem \ref{thm:inequalities}. The $n$ lower bound holds for any ID range, but the $\log \log M$ bound shows that $n$ cannot be the right answer for every ID range: we can get arbitrarily large lower bound if we allow arbitrarily large IDs. The key step in the proof of the lower bounds is to prove that, although a priori a proof of bipartiteness is not required to explicitly give a 2-coloring to the nodes, the nodes can always extract a coloring from the certificate.

\begin{proof} We start with the upper bounds. The $n \log M$ upper bound comes from the certificate made by concatenating, the couples (ID,local proof) for every node, as in Theorem \ref{thm:inequalities}. The $M$ upper bound comes from another proof strategy: the prover will provide a table with $M$ cells, where cell $i$ will contain the bit indicating the colour of the node with ID $i$. In both cases the nodes will get their own colors and the colors of their neighbors from the certificate, and they cannot be fooled by the prover.

We now prove the lower bounds for restricted case of cycles. A priori, for an arbitrary scheme, the prover is not forced to provide explicitly a coloring to the nodes. We show that actually the proof always define a type of coloring. As in section \ref{sec:alos}, we will use blocks of nodes to build a large a number of instances.
The blocks are paths of $2r+1$ nodes, with consecutive IDs from $k(2r+1)+1$ to $(k+1)(2r+1)$. The IDs imply an orientation of the blocks. A \emph{block-based cycle} is a cycle made by concatenating blocks keeping a consistent orientation.

\begin{lemma}\label{lem:2-colouring-blocks}
For every certificate $c$, there exists a function $f_c:[1,M]\mapsto \{0,1\}$, such that for every block-based cycle $H$, if the nodes accept with certificate $c$, then $f$ defines a proper coloring of $H$.
\end{lemma}

In other words, for every certificate $c$, we can define a coloring that is consistent with all the block-based cycles that are accepting $c$. This means that, whatever the scheme is, every node can deduce from the certificate its colour and the ones of its neighbor, and check that there is no colour conflict.

\begin{proof}[Proof (Lemma \ref{lem:2-colouring-blocks}).]
First, note that as the blocks have odd length, a block-based cycle has even length if and only if it is composed of an even number of blocks. Then, for block-based cycles, considering the blocks as vertices, and trying to 2-colour the resulting cycle is equivalent to 2-colour the nodes of the original cycle.

Fix a certificate $c$. Consider the directed graph $G_c$, whose nodes are the blocks, and whose edge are defined the following way. There is an edge $(b_i,b_j)$ if and only if there exists a block-based cycle in which the block $b_i$ is followed by the block $b_j$, and for which $c$ is an accepting certificate. We show two properties of this graph.

\begin{claim}\label{clm:no-odd-cycle}
The graph $G_c$ contains no directed odd cycle.
\end{claim}

Suppose the graph $G_c$ contains a directed odd cycle. We study the block-based cycle $C$, associated to this cycle. This instance is not in the language because it is a cycle of odd length. We claim that the nodes will anyway accept this instance if they are given the certificate $c$.
Consider a node $v$ of this instance, and suppose without loss of generality that it is in the first half of its block $b$, that is its ID is between $k(2r+1)+1$ and $k(2r+1)+r+1$.
The view of $v$, that is its $r$-neighborhood, contains only nodes from its own block $b$ and from the previous block $b'$. By construction of the instance, if there is an edge $(b',b)$ in $G_c$, then there exists a block-based cycle of even length $C'$ where the nodes of $b$ follow the ones of $b'$, and in this instance every node accepts. The view of $v$ in~$C$ is the same as the view of the node with the same ID in $C'$, then it should accept in $C$. Hence every node accepts in $C$; which is a contradiction. Thus the graph $G_c$ contains no directed odd cycle.

\begin{claim}\label{clm:strongly-connected}
Every connected component of the graph $G_c$ is strongly connected.
\end{claim}

Consider the following way to build $G_c$: take an arbitrary ordering of the cycles that accept~$c$, and add them to $G_c$ (i.e. add their edges) one by one. We show the strong connectivity of the connected components by induction. The property holds for the empty graph. Suppose every connected component is strongly connected until some step, and that we add a new cycle. Either it defines a new connected component, which is then strongly connected as it is a directed cycle, or it is added to an existing connected component. And this later case, let $v$ be a node that is in both the old connected component and the new cycle. Note that every node in the modified connected component has a path to $v$ and a path from $v$, this shows that this component is strongly connected.

It is known that a strongly connected digraph with no odd length cycles is bipartite (see e.g. Theorem 1.8.1 in \cite{Bang-JensenG02}). Thus, from Claim \ref{clm:no-odd-cycle} and Claim \ref{clm:strongly-connected}, we get that $G_c$ is bipartite, and we can 2-colour it. As noted at the beginning of this proof, this implies that every even-cycle has a 2-coloring. And the lemma follows.
\end{proof}

Consider the following table. The columns are indexed by the blocks, thus there are $M/(2r+1)$ of them. The rows are indexed by the certificates. The cell that corresponds to block $b$ and certificate~$c$ contains the colour, 0 or 1, given by $f_c$ to the centre node of~$b$. We will now give two simple properties of this table that will imply the two lower bounds.

Let a \emph{balanced bit string} be a bit string with the same number of 0s and 1s. Let the \emph{complement} of a bit string be the bit string where every 1 has been replaced by a 0, and vice versa.

\begin{lemma}\label{lem:balanced}
For every balanced bit string $s$ of length $n$, there exists a raw of the table such that the $n$ first cells form either $s$ or its complement.
\end{lemma}

\begin{proof}[Proof (Lemma \ref{lem:balanced}).]
Consider a balanced bit string $s$. Consider a cycle $H$ made of the $n$ first blocks, in an ordering such that for every $i$ between 1 and $n$, coloring block $i$  with the $i^{th}$ bit of $s$, defines a proper coloring of the cycle. Note that, as $s$ is balanced, such a cycle must exist. This cycle $H$ has even length thus it belongs to the language and there exists an accepting certificate $c$. The $n$ first cell of the row of $c$ must describe a proper coloring of $H$, and there are only two such colorings: $s$ and its complement.
\end{proof}

As for every balanced string of length $n$ there exists a row that matches it or its complement (on the $n$ first cells), and that a row can only correspond to one such string, up to complement, the table must have at least $2^n/2$ rows. This means that there are at least $2^n/2$ different certificates, thus the certificate size is lower bounded by~$n$, up to multiplicative constants.

\begin{lemma}\label{lem:two-columns}
Not two columns of the table can be equal.
\end{lemma}

\begin{proof}[Proof (Lemma \ref{lem:two-columns})]
Suppose columns $i$ and $j$ are equal. Consider an even-length block-based cycle $C$, where the blocks $i$ is linked to the block $j$. Such a cycle always exists. For every certificate $c$, the same colour is given to both blocks $i$ and $j$, because the cell that correspond to ($c$, $b_i$) and ($c$, $b_i$) contain the same bit. Then at least one node of these blocks will reject. This is a contradiction because $C$ is in the language. Thus all the columns are different.
\end{proof}

If we want the columns to be different, we need them to be different bit strings. As there are $2^k$ bits strings of length $k$, we need order of $\log(M)$ different certificates. Then the length of a certificate if at least $\log \log(M)$, up to multiplicative constants. This finishes the proof of Theorem~\ref{thm:bipartite}.
\end{proof}

\section{Local decision and communication complexity}\label{sec:complexity-theory}

In this section we present the nondeterministic hierarchies for local decision and communication complexity.

\paragraph{Nondeterministic hierarchy of local decision.}
Feuilloley et al.\ \cite{FeuilloleyFH16} introduced a nondeterministic \emph{hierarchy of local decision}. It is the distributed computing analogue of the classical polynomial hierarchy. A prover and a disprover take turns, providing each node with a label of size $O(\log n)$. The nodes then look at their constant-radius neighborhood, including the nondeterministic labels, and decide whether they accept or not.

The classes $\Sigma_k$ and $\Pi_k$ correspond to the languages that can be decided using $k$ levels of nondeterminism --- in $\Sigma_k$ the prover goes first, and in $\Pi_k$ the disprover. Let $\ell_1, \ell_2, \dots, \ell_k$ denote the $k$ levels of nondeterministic labels provided to the nodes. A language $L \in \Sigma_k$ if and only if there exists a verifier $A$ such that
\[
  (G,x) \in L \iff \exists \ell_1, \forall \ell_2, \dots, \operatorname{Q} \ell_k, \forall v \in V(G), A \text{ accepts.}
\]
Here $\operatorname{Q}$ denotes the existential quantifier if $k$ is odd and the universal quantifier otherwise. The classes $\Pi_k$ are defined similarly, but with the disprover (i.e. universal quantifier) going first.

The classes that corresponds to the disprover talking last collapse to the previous level, and the only interesting levels are $\Sigma_1, \Pi_2, \Sigma_3, \dots$, which are called $(\Lambda_k)_k$. The complements of these classes are denoted by $\colambda{i}$ and we have that $\colambda{k} \subseteq \Lambda_{k+1}$~\cite{FeuilloleyFH16}, i.e., decision can always be reversed using an extra quantifier with $O(\log n)$ bits. As shown in Theorem~\ref{thm:alos}, in general, $\Omega(\log n)$ bits are also required for reversing decision.

The main open question of Feuilloley et al.\ \cite{FeuilloleyFH16} was whether $\Lambda_2$ and $\Lambda_3$ were different or not. As in the polynomial hierarchy, the equality $\Lambda_k = \Lambda_{k+1}$ of two levels would imply a collapse of the local hierarchy down to the $k$th level. We show that this question is related to long-standing open questions nondeterministic communication complexity~\cite{babai86complexity}.

\paragraph{A hierarchy for global certificates.} Similar to the hierarchy of local certificates, we can define a hierarchy for the global certificates. Define $\Sigma^G_k$ and $\Pi^G_k$ as previously, except that the labels $\ell_1, \ell_2, \dots, \ell_k$ are global certificates seen by all nodes.

\paragraph{Communication complexity.} We will compare the hierarchies of nondeterministic local decision to the hierarchy of nondeterministic communication complexity defined by Babai et al.\ \cite{babai86complexity}.

In the communication complexity setting we are given a boolean function $f$ on $2n$ bits. Two entities, Alice and Bob, are each given $n$-bit vectors $x$ and $y$, and have to decide if $f(x \cup y) = 1$. They can communicate through a reliable channel and have unlimited computational resources. The measure of complexity is the number of bits Alice and Bob need to communicate in order to decide $f$. For more details, see for example the book~\cite{KushilevitzN97}.

In nondeterministic communication complexity Alice and Bob have access to nondeterministic advice (we will say that it is given by a \emph{prover}). The cost of a protocol is the sum of the number of bits communicated by Alice and Bob and the number of advice bits given by the prover. This means that messages of Alice and Bob can equivalently be encoded in the advice.

Babai et al.\ defined a hierarchy of nondeterministic communication complexity~\cite{babai86complexity}. In addition to Alice and Bob we have two players, whom we will call \emph{prover} and \emph{disprover} for consistency, giving nondeterministic advice to Alice and Bob. Prover and disprover will alternate $k$ times and each time give an advice string of $g(n)$ bits. Now we define the class $\Sigcc_k(g(n))$ of boolean functions as the set of functions such that there exists an algorithm $A$ for Alice, and an algorithm $B$ for Bob such that if $f \in \Sigcc_k(g(n))$, then

\[
  \forall x,y, \exists \ell_1, \forall \ell_2, \dots, \operatorname{Q} \ell_k A(\ell_1, \ell_2, ...,\ell_k, x) = B(\ell_1, \ell_2,...,\ell_k,y)= 1 \iff f(x,y)=1.
\]

Again $\operatorname{Q}$ denotes the existential quantifier if $k$ is odd and the universal quantifier otherwise. The classes $\Picc_k(g(n))$ are defined similarly, but with the disprover going first.
We are particularly interested in this hierarchy when $g(n) = O(\log n)$. Note that in their work, Babai et al.\ consider the hierarchy for $g(n) = O(\poly(\log n))$~\cite{babai86complexity}.

\subsection{Connecting local decision and communication complexity}

In this section we partially formalize the intuition that complexity of local verification is connected to communication complexity. We show that general lower bound proof techniques for nondeterministic local verification will also apply to communication complexity. We then show that if one considers global proofs instead of local ones, the result can be strengthened.

\begin{theorem} \label{thm:cc-lv-connection}
  For every boolean function $f$, there exists a distributed language $L_f$ such that if $f \in \Sigma_k^{cc}(g(n))$ for odd $k$ or $f \in \Pi_k^{cc}(g(n))$ for even $k \geq 2$, then $L_f \in \Lambda_k(g(n))$.
\end{theorem}

The proof is by showing that there exists a family of languages such that a nondeterministic verification scheme can simulate a nondeterministic communication protocol. The theorem partially explains why it is difficult to separate the different levels of the local decision hierarchy --- the question is inherently tied to long-standing open questions in communication complexity~\cite{babai86complexity}.

\begin{proof}[Proof of Theorem~\ref{thm:cc-lv-connection}]
  Let $f$ be a boolean function on $2n$ variables. We will construct an infinite family of graphs $\mathcal{G}_n = \bigl(G(n,t,x,y)\bigr)_{t,x,y}$ and a related language $L_f$.

  The graph $G_{(n,t,x,y)}$ consists of a path $P_{2t+1} = (v_1, v_2, \dots, v_{2t+1})$ of length $2t+1$, and two sets of nodes, $V_A$ and $V_B$ of size $n$. Let us denote $v_A = v_1$ and $v_B = v_{2t+1}$. We add an edge between each $v \in V_A$ and $v_A$, and an edge between each $u \in V_B$ and $v_B$. The nodes $v_A$ and $v_B$ are labelled with their respective identities.

  Parameters $x$ and $y$ are bit vectors of length $n$, corresponding to the inputs of players $A$ and $B$ in the communication complexity setting. To encode the input vectors, we use graphs on $V_A$ and $V_B$, respectively. There are $2^n$ possible input vectors. We'll define a function $\phi$ that maps each graph on $n$ nodes to an $n$-bit vector. Since the encoding of the input cannot depend on the unique identifiers, $\phi$ must map all graphs of the same isomorphism class to the same vector. Finally, since there are at least
  \[
    2^{n \choose 2} / n! = \Omega(2^{n^2})
  \]
  such graph isomorphism classes, we can find a $\phi$ such that for all $x \neq y$, we have that $\phi^{-1}(x) \cap \phi^{-1}(y) = \emptyset$.

   Given $\phi$, $x$, and $y$, we can choose two graphs $G_A \in \phi^{-1}(x)$ and $G_B \in \phi^{-1}(y)$, identify the node sets $V_A$ and $V_B$ with $V(G_A)$ and $V(G_B)$, respectively, and add the corresponding edges to the graph $G(n,t,x,y)$. We will use $G_A$ and $G_B$, respectively, to denote these graphs on node sets $V_A$ and $V_B$. Nodes $v_A$ and $v_B$ are labelled as special nodes so that the structure of $G_A$ and $G_B$ can be detected. We denote this graph construction by $G(n,t,x,y)$.

   \paragraph{Local verification of $\mathcal{G}_f$.} A single $O(\log n)$-bit certificate is enough to verify the structure of $G(n,t,x,y)$. It first consists of a spanning tree of $P_{2t+1}$: node $v_A$ is marked as root, and each node $v_i$ has a pointer to $v_{i-1}$ and a counter $i$, its distance to the root. It also contains the value $n$. The nodes $v_A$ and $v_B$ can check that the sizes of the graph $G_A$ and $G_B$ are consistent with this value. They also check that there are no other outgoing edges from $G_A$ and $G_B$. Nodes $v_A$ and $v_B$ can see all nodes of $G_A$ and $G_B$, and determine their isomorphism classes, and compute $x = \phi(G_A)$ and $y = \phi(G_B)$.

   \paragraph{Deciding $L_f$.} We say that $G \in L_f$ if and only if
   \begin{enumerate}[noitemsep]
     \item the structure of $G$ is that of $G(n,t,x,y)$ for some setting of the parameters, and
     \item the function $f$ evaluates to 1 on $\phi(G_A) \cup \phi(G_B)$.
   \end{enumerate}

   Now assume that $f$ is on the $k$th level of the communication complexity hierarchy with $s = \Omega(\log n)$ bits of nondeterminism. We can use this implied protocol $P$ to solve $L_f$ on the $k$th level. If the graph structure is correct, the prover and disprover essentially simulate their counterparts from the communication complexity setting, and label \emph{all} nodes on $P_{2t+1}$ as if in $P$. Then $v_A$ can simulate $A$ and $v_B$ can simulate $B$, accepting if and only if $f(x,y) = 1$. If the prover tries to deviate from this strategy, nodes can see that its labelling of $P_{2t+1}$ is not constant, and reject. If the disprover tries to deviate, the prover can construct a certificate pointing to this error, and all nodes will accept.
   \end{proof}

\paragraph{Global proofs and communication complexity.}

In the setting of global proofs we can show a slightly stronger theorem.

\begin{theorem} \label{thm:cc-gv-collapse}
  For every boolean function $f$ and every $g(n) = \Omega(\log n)$ there exists a distributed language $L_f$ such that $L_f \in \Lambda_k(g(n))$, for $k \ge 1$ if and only if $f$ is in the $k$th level of the communication complexity hierarchy with $O(g(n))$ bits of nondeterminism, in particular $f \in \Sigcc_k$ for $k$ odd or $f \in \Picc_k$ for $k$ even.
\end{theorem}

In particular, this theorem implies that any collapse in the hierarchy for global certificates implies a collapse in the corresponding communication complexity hierarchy.

\begin{proof}[Proof of Theorem~\ref{thm:cc-gv-collapse}]
  We show that with respect to the language $L_f$ defined in the proof of Theorem~\ref{thm:cc-lv-connection}, the communication complexity model and the global verification model can simulate each other.

  \begin{enumerate}
    \item \emph{Communication protocol implies a global verification protocol.} The proof proceeds essentially as in the proof of Theorem~\ref{thm:cc-lv-connection}. Using $O(t \log n)$ bits the global certificate can give the list of nodes on the path between $v_A$ and $v_B$. If a node has degree 2, it must see its own name on this list. Nodes $v_A$ and $v_B$ can again locally verify the structure of $G_A$ and $G_B$ and recover $x$ and $y$. Finally the prover and disprover follow the communication protocol $P$ on instance $(x,y)$, allowing nodes $v_A$ and $v_B$ to simulate Alice and Bob.

    \item \emph{Global verification scheme implies a communication protocol.} Assume there is a $k$th level global verification scheme with $g(n)$-bit certificates for $L_f$.

    Alice and Bob will simulate this scheme as follows. Construct a virtual graph $G(x,y)$ consisting of three parts: the nodes $v_A$ and $v_B$, a path $P_{2t+1}$ of length $2t+1$ between them, and graphs $H(x)$ and $H(y)$ that are the first elements (in some order) of $\phi^{-1}(x)$ and $\phi^{-1}(y)$, respectively. Finally, all nodes of $H(x)$ are connected to $v_A$ and all nodes of $H(y)$ to $v_B$. Only Alice will know $H(x)$ and only Bob $H(y)$.

    This graph is in $L_f$ if and only if $f(x,y) = 1$: the structure is exactly as in the definition of $L_f$.

    Now the nondeterministic prover and disprover can simulate their counterparts in the global verification scheme. Alice and Bob accept if and only if the prover can force all nodes they control to accept. Thus the complexity is bounded by the complexity $g(n)$ of the global verification scheme.
  \end{enumerate}
\end{proof}

\paragraph{Acknowledgement.} The authors would like to thank Jukka Suomela for mentioning that \textsc{Amos} could be an interesting problem in this context.

\newpage{}

\DeclareUrlCommand{\Doi}{\urlstyle{same}}
\renewcommand{\doi}[1]{\href{http://dx.doi.org/#1}{\footnotesize\sf doi:\Doi{#1}}}

\bibliographystyle{plainnat}
\bibliography{local_verification}

\end{document}